\documentclass[submission]{eptcs}
\providecommand{\event}{EXPRESS 2012} % Name of the event you are submitting to
\usepackage{amsmath,amssymb,stmaryrd,xspace,tikz,ifthen,subfigure,paralist}
\usetikzlibrary{arrows,automata,shapes,snakes,backgrounds}

\newtheorem{defi}{Definition}
\newtheorem{theo}{Theorem}
\newtheorem{prop}{Proposition}
\newtheorem{lemm}{Lemma}
\newtheorem{coro}{Corollary}

\newenvironment{definition}{\begin{defi} \rm }{\end{defi}}
\newenvironment{theorem}{\begin{theo} \rm }{\end{theo}}
\newenvironment{proposition}{\begin{prop} \rm }{\end{prop}}
\newenvironment{lemma}{\begin{lemm} \rm }{\end{lemm}}

\newenvironment{proof}{\begin{trivlist} \item[\hspace{\labelsep}\bf Proof:]}{\hfill$\Box$\end{trivlist}}

\newcommand{\AP}{\ensuremath{\mathsf{AP}}}
\newcommand{\Lit}{\ensuremath{\mathsf{Lit}}}

\renewcommand{\omit}[1]{}

\usepackage{environ}

\newcommand{\sem}[1]{\ensuremath{[\![ #1 ]\!]}}

\def\rmust{R^{+}}

\def\mixsim{\leq_{mix}}

\def\eqdef{\triangleq}

\def\iff{\Leftrightarrow}

\def\iffdef{\stackrel{def}\Leftrightarrow}

\def\implies{\Rightarrow}

\def\nat{\mathbb{N}}

\def\false{\bot}
\def\true{\top}
\def\vars{\mc{V}}

\newcommand{\diamop}{\lozenge}
\newcommand{\boxop}{\Box}

\newcommand{\mc}[1]{\ensuremath{\mathcal{#1}}}
\newcommand{\ie}{i.e.\xspace}
\newcommand{\eg}{e.g.\xspace}

\newcommand{\ap}{{:}}

\newcommand{\rel}[1]{\ensuremath{\mathrel{#1}}}

\newcommand{\KMTS}{KMTS\xspace}
\newcommand{\MTSs}{KMTSs\xspace}
\newcommand{\KMTSs}{KMTSs\xspace}
\newcommand{\GTS}{GTS\xspace}
\newcommand{\GTSs}{GTSs\xspace}

\newcommand{\concret}[3]{\ensuremath{\mathbb{C}_{#3}[\langle #1,#2 \rangle]}}
\newcommand{\concretfix}[4]{\ensuremath{\mathbb{C}_{#4}[\langle #1,#2 \rangle,#3]}}

\newcommand{\spair}[2]{\ensuremath{\langle #1, #2 \rangle}}

\def\calm{{\cal M}}
\def\calf{{\cal F}}
\def\calc{{\cal C}}

\def\mKS{\textbf{KS}}
\def\mMTS{\textbf{KMTS}}
\def\mKMTS{\textbf{KMTS}}
\def\mGTS{\textbf{GTS}}

\def\lmu{{\cal L}_{\mu}}

\def\absmodels{\ensuremath{\models^{\alpha}}}

\def\metaref{\preceq}

\newcommand{\states}[1]{\textbf{states}(#1)}
\newcommand{\tuple}[1]{\langle #1 \rangle}

\begin{document}

\title{Expressiveness and Completeness in Abstraction}
\author{Maciej Gazda and Tim A.C. Willemse
\institute{Department of Computer Science, Eindhoven University of Technology (TU/e),
\\ {P.O.~Box~513}, NL-5600~MB~~Eindhoven, The Netherlands
}
}

\def\titlerunning{Expressiveness and Completeness in Abstraction}
\def\authorrunning{Maciej Gazda and Tim A.C. Willemse}
\maketitle

\begin{abstract}
We study two notions of expressiveness, which have appeared in
abstraction theory for model checking, and find them incomparable in
general. In particular, we show that according to the most widely
used notion, the class of Kripke Modal Transition Systems is strictly
less expressive than the class of Generalised Kripke Modal Transition
Systems (a generalised variant of Kripke Modal Transition Systems
equipped with hypertransitions). Furthermore, we investigate the ability
of an abstraction framework to prove a formula with a finite abstract
model, a property known as \emph{completeness}. We address the issue of
completeness from a general perspective: the way it depends on certain
abstraction parameters, as well as its relationship with expressiveness.

\end{abstract}

\section{Introduction}

Model checking \cite{Clarke:2000:MC:332656} is one of the key technologies
for formal software verification. Given a model of a program or a
process and a specification of the required behaviour in the form of a
logical formula, a model checker can automatically verify whether the
model satisfies the specification. A model checker typically explores
the entire state space of a program. Such a state space is enormous in
most practical applications.

Abstract interpretation \cite{CC:77,ClarkeGL92} is among the most
important techniques designed to handle the state space explosion problem,
making many instances of the model checking problem tractable. It works
by approximating the artefacts of the original model, the so-called
\emph{concrete model}, by simpler \emph{abstract objects}. A model
transformed in this way has a smaller abstract state space.  The loss of
detail in this model may allow a model checker to successfully verify the
property, but it can also give rise to an inconclusive answer.  The cause
of the inconclusive answer may be resolved by successive refinements of
the abstraction \cite{CGJLV:03}, leading to finer-grained abstract models.

Assuming that, as most works on abstraction do, concrete models are
modelled by \emph{Kripke Structures}, we investigate two key properties
of abstraction formalisms for Kripke Structures. First, we study the
\emph{expressiveness} of the formalism. This gives us the information
about the classes of concrete structures that can be described by abstract
models. Second, we study the \emph{completeness} of the formalism. In
abstraction, completeness is the degree to which properties
of a concrete model can be proved using a \emph{finite} abstraction.

A systematic survey of the literature reveals that there is an
abundance of different abstraction formalisms for Kripke Structures.
Kripke Structures equipped with the usual simulation relation themselves
form one of the first studied abstraction formalisms, but their power
is rather limited.  Below, we list the most important families of abstraction
formalisms:
\begin{enumerate}
 \item \emph{Modal Transition Systems} (MTSs) \cite{LT:88} with \emph{may}
 and \emph{must} transitions and a built-in consistency requirement,
 and related formalisms, see \eg~\cite{GJ:03}, such as \emph{Kripke
 Modal Transition Systems} (\MTSs) \cite{HJS:01}.

 \item \emph{Mixed Transition Systems} (MixTSs) \cite{DGG:97}, a modelling
 formalism similar to \MTSs, but with the added capability of expressing
 inconsistent specifications.

 \item \emph{Generalised Kripke Modal Transition Systems} (\GTSs)
 \cite{SG:04} with \emph{must hypertransitions}; similar structures were
 already introduced by Larsen and Xinxin in~\cite{LX:90} under the name
 of Disjunctive Transition Systems \cite{LX:90}, although there, these
 structures served a different purpose.

 \item Tree Automata (TA) \cite{DN:05}, the most expressive and complete
 of all of the listed formalisms.

\end{enumerate}

In this paper, we mostly restrict ourselves to GTSs and KMTSs. This is
because in practice, GTSs and KMTSs are the key abstraction formalisms
used. Tree Automata, while being the ``most complete'' among the
abstraction formalisms, are mainly of theoretical importance due
to the complexity of computing an abstraction using this formalism.

Expressiveness of abstraction formalisms has been studied before.
In fact, Wei \emph{et al.}~\cite{WGC:11} proved that the
formalisms from the KMTS family with may and must transitions have the same
expressiveness as GTSs with must hypertransitions; they claim to:
\begin{quote}
``...complete the picture by showing the expressive equivalence \emph{between}
these families.''
\end{quote} 
At first glance, this seems rather odd: the GTS abstraction formalism
is more liberal than any member of the KMTS family of abstraction
formalisms. While the arguments in~\cite{WGC:11} are sound, a closer
inspection of their results reveals that their notion of expressiveness
can be regarded as non-standard. It is therefore not immediately clear
whether their results are comparable to the expressiveness results
reported by, \eg\ Godefroid and Jagadeesan in~\cite{GJ:03}. 

We show that the notion of expressiveness \emph{does} make a difference:
using the notion used by \eg\ Godefroid and Jagadeesan we conclude that
the GTS abstraction formalism is \emph{strictly more expressive} than
members of the KMTS family of abstraction formalisms. The expressiveness
results claimed by Wei \emph{et al} are therefore likely to become a
source of confusion.  We henceforth refer to the notion of expressiveness
used by Wei \emph{et al} as \emph{contextual expressiveness}.

The aforementioned paper \cite{WGC:11} suggests the expressiveness and completeness of an
abstraction formalism are closely related, given the following quote:
\begin{quote}
``The work of Godefroid and Jagadeesan, and Gurfinkel and Chechik showed
that the models in the KMTS family have the same expressive power and
are equally precise for SIS. Dams and Namjoshi  showed that the three
families considered in this paper are subsumed by tree automata. We
completed the picture by proving that the three families are equivalent as
well. Specifically, we showed that KMTSs, MixTSs and GKMTSs are relatively
complete (in the sense of [Dams and Namjoshi]) with one another.''
\end{quote}
Since Dams and Namjoshi's paper~\cite{DN:04} to which they refer only studies
completeness, we can only conclude that Wei \emph{et al} consider completeness and expressiveness
as equivalent notions. Since they are defined in
different ways, it is therefore interesting to know the exact relationship
between the two, if one exists.

We have been able to establish only a weak link between expressiveness and
completeness, namely, only when assuming a \emph{thorough semantics} for logical
formulae, more expressive abstraction formalisms are more complete. Since
in~\cite{WGC:11}, thorough semantics is not used, our findings are not
in support of the claims in the above quote. 

Finally, we investigate the notion of completeness itself in more detail.
For instance, it is known that GTSs are complete for the fragment of
least-fixpoint free $\mu$-calculus formulae~\cite{AGJ:04}. However,
it is not known whether those are the only formulae for which GTSs are
complete. Our investigations reveal that the answer to this open question
is ambiguous and depends on the setting that is used.

\paragraph{Related work.}

Expressiveness of modelling formalisms for abstraction has been
first studied by Godefroid and Jagadeesan in \cite{GJ:03}. There,
it has been proved that Partial KSs (Kripke structures with possibly
unknown state labels), MTSs (without state labels) and KMTSs are
equally expressive. The proof consists of defining 3 translations
that preserve a variant of mixed simulation called a `completeness
preorder'. In~\cite{GC:05}, Gurfinkel and Chechik have shown that
Partial classical Kripke Structures (with each atomic proposition either
always ``true'' or ``false'', or always ``maybe'' ) are expressively
equivalent to the above formalisms. Subsequently, Wei, Gurfinkel and
Chechik \cite{WGC:11} have studied expressiveness in the context of a
fixed abstraction; we come back to this notion in more detail in
Section~\ref{sec:expressiveness}.

Expressiveness of various modelling formalisms in the context of refinement
has been studied in \cite{FE:09}. In contrast to our work, the authors
consider only deterministic structures as proper concretisations of
a model.

Dams and Namjoshi were the first to explicitly address the question
whether there are abstraction frameworks that are complete for the entire
$\mu$-calculus. They answer this question in~\cite{DN:04} by introducing
abstraction based on \emph{focused transition systems}. In their follow-up
work~\cite{DN:05}, they show that these focused transition systems are
in fact variants of $\mu$-automata, enabling a very brief and elegant
argument for completeness of their framework.

The \GTS/DMTS framework has received a considerable interest from the
abstraction community.  Shoham and Grumberg studied the precision of the
framework in~\cite{SG:08}; Fecher and Shoham, in~\cite{FS:11} used the
framework for a more algorithmic approach to abstraction, by performing
abstraction in a lazy fashion using a variation on parity games.

\paragraph{Outline.} In Section~\ref{sec:preliminaries}, we introduce
the abstraction formalisms and the basic mathematical machinery
needed to understand the remainder of the paper. We investigate the
expressiveness of the GTS and KMTS abstraction formalisms 
in Section~\ref{sec:expressiveness}. We then proceed to study completeness
in Section~\ref{sec:completeness}; there, we provide a formal framework
that allows to compare completeness of different formalisms, we study the
relation between completeness and expressiveness and we more accurately
characterise the set of formulae for which GTSs are
complete. We wrap up with concluding remarks and issues for future work
in Section~\ref{sec:conclusions}.

\def\atprop{\ensuremath{\textsf{AP}}}
\def\lit{\ensuremath{\textsf{Lit}}}

\section{Preliminaries}
\label{sec:preliminaries}

The first basic ingredient of an abstraction theory is the class of
concrete structures ${\cal C}$, representing the objects (programs,
program models) that we wish to analyse. Throughout this paper, 
we restrict ourselves to the
setting where ${\cal C}$ consists of all Kripke Structures (KSs), or possibly some subclasses of KSs (\eg the set of finitely branching Kripke Structures). 

Let $\atprop$ denote an arbitrary, fixed set of \emph{atomic
propositions}, used to specify properties of states; propositions and
their negations constitute the set of \emph{literals} $\lit = \atprop
\cup \{\neg p \,\mid\, p \in \atprop\}$.  Below, we recall the definition
of Kripke Structures.

\begin{definition} 
 A Kripke Structure is a tuple $\langle S, S^0, R, L \rangle$ where:
 \begin{itemize}
\item $S$ is a set of states, 
\item $S^0 \subseteq S$ is a set of initial states
\item $R \subseteq S \times S$ is the transition relation;

\item $L {:} S \to 2^\Lit$ is a labelling function such that $L(s)$
contains exactly one of $p$ and $\neg p$ for all $p \in \AP$. 
\end{itemize}
The class of all Kripke Structures is denoted $\mKS$.
\end{definition}

The concrete structures are described using abstract models; as a
convention, we use symbols ${\cal M}_{(i)}$ to denote classes of abstract
models and $M_{(i)}$ for instances of models. We assume that models consist
of states with a distinguished set of initial states, and additional
structural components such as transitions and labels. Formally, models are
tuples of the form $M = (S,S^0,\Sigma)$, where $\Sigma$ represents the
aforementioned structural artefacts. For a given model $M$, $\states{M}$
is the set of states underlying $M$. The notation $\tuple{M,s}$ will be
used to represent the state $s$ of a model $M$. We also make the general
assumption that abstract models considered in this paper are finite.\\

Properties of concrete and abstract models are expressed using a certain logic $L$. In case of concrete structures, we will use $\mu$-calculus ($\lmu$) with its standard semantics. The same logic is
used for our abstract models; however, for such models, semantics of $\lmu$,
given by the definition of the satisfaction relation by $\models^{\alpha}$,
may vary (\eg\ inductive or thorough semantics). Below, we first present
the syntax of $\lmu$ and its semantics for Kripke Structures.

\begin{definition} A $\mu$-calculus formula (in positive form) is a formula
generated by the following grammar:
\[
\varphi,\psi ::=  \true \mid \false \mid l \mid X \mid 
\varphi \wedge \psi \mid \varphi \vee \psi \mid \Box \varphi \mid \lozenge \varphi \mid
\nu X.~\varphi \mid \mu X.~\varphi
\]
where $l \in \Lit$ and $X \in \mc{V}$ for a set of propositional variables
$\mc{V}$. The $\mu$ and $\nu$ symbols denote the least and greatest fixpoint respectively. The semantics of a formula $\varphi$ is an inductively
defined  function $\sem{\_}$, in the context of a Kripke Structure  $M = \langle S, S_0, R, L \rangle$ and an environment
$\eta \ap \mc{V} \to 2^S$ assigning sets of states to the propositional
variables:
\[
\begin{array}{lclp{1cm}lcl}
\sem{\true}{^\eta} &=& S & & \sem{\false}{^\eta} &=& \emptyset\\

\sem{l}{^\eta} &=& \{s \in S ~|~ l \in L(s)\}  & & \sem{X}{^\eta} &=& \eta(X) \\

\sem{\varphi \wedge \psi}{^\eta} &=& \sem{\varphi}{^\eta} \cap \sem{\psi}{^\eta} & & \sem{\varphi \vee \psi}{^\eta} &=& \sem{\varphi}{^\eta} \cup \sem{\psi}{^\eta} \\ 

\sem{\lozenge \varphi}{^\eta} &=&  \tilde\lozenge \sem{\varphi}{^\eta} & & \sem{\Box \varphi}{^\eta} &=&  \tilde\Box \sem{\varphi}{^\eta}  \\

\sem{\mu X. \varphi}{^\eta} &=&  \mu U. \sem{\varphi}{^{\eta[X := U]}} & & \sem{\nu X. \varphi}{^\eta} &=&  \nu U. \sem{\varphi}{^{\eta[X := U]}} \\
\end{array}
\]
Here, we used the following two abbreviations: 
\[\begin{array}{lcl}
\tilde\lozenge U = \{s \in S ~|~ \exists q \in S.~ s \rel{R} q \wedge q \in U\} &
\qquad\text{and}\qquad
&
\tilde\Box U = \{s\in S ~|~ \forall q \in S.~ s \rel{R} q \Rightarrow
q \in U \}
\end{array}
\]
In case the formula $\varphi$ is \emph{closed} (\ie, does not contain free
propositional variables: variables that are not bound by a surrounding
fixpoint binding the variable), its semantics is independent of the
environment $\eta$, and we drop $\eta$ from the semantic brackets.
We say that 
a closed formula $\varphi$ is \emph{true} in a state $s \in S$, denoted
$\tuple{M,s} \models \varphi$ if $s \in \sem{\varphi}$; $\varphi$ is
\emph{false} in $s$, denoted $\tuple{M,s} \not\models \varphi$ if $s
\notin \sem{\varphi}$.

\end{definition}
The $\mu$-calculus is a highly expressive logic, subsuming
temporal logics such as LTL, CTL and CTL$^*$. It is capable of
expressing safety properties, liveness properties and (complex)
fairness properties. Typically, least fixpoint subformulae express
\emph{eventualities}, whereas greatest fixpoint subformulae express
\emph{invariance}. For instance, the formula $\nu X. (\Box X \wedge l)$
expresses that $l$ holds invariantly on all computation paths, whereas
$\mu X. ((\Box X \wedge \lozenge \top) \vee l)$ expresses that property
$l$ holds eventually on all computation paths. By mixing least and
greatest fixpoints one can construct (computationally and intuitively)
more complex formulae such as $\nu X. \mu Y. ((\lozenge X \wedge l)
\vee \lozenge Y)$, which expresses that there is a computation path on
which $l$ holds infinitely often.\\

An essential component of an abstraction formalism
is the notion of \emph{description}, \emph{approximation} or
\emph{refinement}, which provides a link between concrete and abstract
\emph{structures}. Typically, there is a ``meta-relation'' between classes
of structures (here denoted with $\metaref$). Such a meta-relation is
typically a relation on states or languages; for instance, \emph{mixed
simulation} or \emph{language containment}. 

As already mentioned, the final ingredient of an abstraction formalism is the semantics of the
logic on abstract models. The abstract $\lmu$ semantics should
include the definition of both \emph{satisfaction} (denoted with
$\models^{\alpha}$) and \emph{refutation} $\not \models^{\alpha}$
predicates. Unlike in Kripke Structures, because of the loss of precision
it might be the case that neither $M \models^{\alpha} \varphi$ nor $M
\not \models^{\alpha} \varphi$ is true for an abstract model $M$. However,
we assume that the models are \emph{consistent}, \ie for no model
$M$ both $M \models^{\alpha} \varphi$ and $M \not \models^{\alpha}
\varphi$ are true. This can typically be ensured by imposing a certain
syntactic restriction; concretely, in our definitions of GTSs and KMTSs,
we require that all must transitions are matched by the corresponding
may transitions.

Given an $\lmu$ semantics for abstract models, the refinement
relation has to meet a soundness property (also known as a weak
preservation property).  That is, we require that whenever $\spair{K}{s_K}\,\rho\,\spair{M}{s_M}$ holds for some specific states and
relation $\rho$ of type $\metaref$, then $\spair{M}{s_M} \models^{\alpha}
\varphi$ implies $\spair{K}{s_K} \models \varphi$ for every $\varphi \in \lmu$.

Before we give the formal definitions of the GTS and KMTS abstraction
frameworks, we formally define the notion of an \emph{abstraction
formalism}.

\begin{definition} 
\label{def:absform}
We define an \emph{abstraction formalism} ${\cal F}$ as a quadruple
$\tuple{\calc,\calm,\metaref,\models^{\alpha}}$, where $\calc$ is a
class of concrete structures (\ie a subclass of $\mKS$); ${\cal M}$ is
a class of models, $\metaref$ is a class of (structural) refinement
relations between ${\cal C}$ and ${\cal M}$, and $\models^{\alpha}$
specifies the semantics of the logic ($\mu$-calculus) on abstract models.

\end{definition}
\def\absform{\tuple{\calc,\calm,\metaref,\models^{\alpha}}}
Below, we define the class of Generalised Kripke Modal Transition Systems
(\GTSs) \cite{SG:04}. The class of Kripke Modal Transition Systems
(\KMTSs) can be viewed as a specialisation of GTSs; in turn, Kripke
Structures can be considered as a specialisation of \KMTSs.

\begin{definition}
A \emph{Generalised Kripke Modal Transition System} (\GTS) is a tuple
$M = \langle S, S^0, R^+, R^-, L \rangle$ where:
\begin{compactitem}
\item $S$ is a set of states, 
  
\item $S^0 \subseteq S$ is a set of initial states,

\item $R^- \subseteq S \times S$ is the \emph{may} transition relation,
%we require $R^-$ to be total,

\item $R^+ \subseteq S \times 2^S$ is the \emph{must} transition relation;
we require that $s \rel{R^+} A$ implies $s \rel{R^-} t$ for all $t \in A$,

\item $L {:} S \to 2^\Lit$ is a labelling function; we require that $L(s)$
contains at most one of $p$ and $\neg p$ for all $s \in S, p \in \AP$.

\end{compactitem}
The system $M$ is called a \emph{Kripke Modal Transition System}
(\KMTS) if for all $s \in S, A \subseteq S$ for which $s \rel{R^+} A$
we have $|A| = 1$.  The class of all GTSs will be denoted with $\mGTS$
and the class of all \KMTSs is denoted by $\mKMTS$.

\end{definition}
Note than an \KMTS $M$ can be identified with a Kripke Structure, if
for all $s,s' \in S$ we have $s \rel{R^+} \{s'\}$ iff $s \rel{R^-} s'$
and $L(s)$ contains precisely one of $p$ and $\neg p$ for all $p \in \AP$.\\

The intuition behind a must hypertransition $s R^+ A$ is that it is
guaranteed that there is a transition from $s$ to some state in $A$,
but the exact state in $A$ to which this transition leads is not
determined upfront, offering some extra flexibility.  In contrast, the
extra condition that is imposed on \KMTSs ensures that the destination
of a must hypertransition \emph{is} determined.

Next, we formalise the notion of approximation, or refinement, between
concrete and abstract structures. The \emph{de-facto} approximation
relation for GTSs is \emph{mixed simulation}; this approximation relation
also appears under the names ``completeness preorder'' \cite{GJ:03,DN:05}
and ``refinement preorder'' \cite{WGC:11}.  Given that GTSs generalise
Kripke Structures, we define the notion of mixed simulation between
abstract structures only.

\begin{definition}
\label{def:mixsim}
Let $M_1 = \langle S_1, S_1^0, R^+_1, R^-_1, L_1 \rangle$ and
$M_2 = \langle S_2, S_2^0, R^+_2, R^-_2, L_2 \rangle$ be two GTSs. A relation
$H \subseteq S_1 \times S_2$ is a mixed simulation 
if $s_1 \rel{H} s_2$ implies
\begin{itemize}
\item $L_2(s_2) \subseteq L_1(s_1)$,

\item if $s_1 \rel{R_1^-} s_1'$, then there exists
$s_2' \in S_2$ such that $s_2 \rel{R_2^-} s_2'$ and $s_1' \rel{H} s_2'$,

\item if $s_2 \rel{R_2^+} A_2$, then there exists
$A_1 \subseteq S_1$ such that $s_1 \rel{R_1^+} A_1$, and for every
$s_1' \in A_1$ there is some $s_2' \in A_2$ such that
$s_1' \rel{H} s_2'$.

\end{itemize}
We write $\langle M_1, s_1 \rangle \mixsim \langle M_2, s_2\rangle$
if $s_1 \rel{H} s_2$ for some mixed simulation $H$.  Mixed simulation
between models $M_1$ and $M_2$, denoted $M_1 \mixsim M_2$, holds
iff for all initial states $s^0_1 \in S^0_1$ there is a corresponding
initial state $s^0_2 \in S^0_2$ such that $\tuple{M_1,s^0_1} \mixsim
\tuple{M_2,s^0_2}$.

\end{definition}

We proceed to define the \emph{standard inductive semantics} (SIS)
of a $\mu$-calculus formula in the setting of GTSs. This can be
done in more than one way; our definition is taken from 
Shoham and Grumberg~\cite{SG:08}.
\newcommand{\semtrue}[2]{\sem{#1}_{\texttt{tt}}^{#2}}
\newcommand{\semfalse}[2]{\sem{#1}_{\texttt{ff}}^{#2}}

\begin{definition} 
The standard inductive semantics of a formula $\varphi$ is inductively
defined by two functions $\semtrue{\_}{}$ and $\semfalse{\_}{}$, in the
context of a GTS $M = \langle S, R^+, R^-, L \rangle$ and an environment
$\eta \ap \mc{V} \to 2^S$ assigning sets of states to propositional
variables:

\[
\begin{array}{lclp{1cm}lcl}
\semtrue{\true}{^\eta} &=& S &  & \semfalse{\true}{^\eta} &=& \emptyset \\

\semtrue{\false}{^\eta} &=& \emptyset &  & \semfalse{\false}{^\eta} &=& S \\

\semtrue{l}{^\eta} &=& \{s \in S ~|~ l \in L(s)\} &  
   & \semfalse{l}{^\eta} &=& \{s \in S ~|~ \neg l \in L(s)\} \\

\semtrue{X}{^\eta} &=& \eta(X) &  & \semfalse{X}{^\eta} &=& \eta(X) \\

\semtrue{\varphi \wedge \psi}{^\eta} &=& \semtrue{\varphi}{^\eta} \cap \semtrue{\psi}{^\eta} &  
   & \semfalse{\varphi \wedge \psi}{^\eta} &=&\semfalse{\varphi}{^\eta} \cup \semfalse{\psi}{^\eta}\\

\semtrue{\varphi \vee \psi}{^\eta} &=& \semtrue{\varphi}{^\eta} \cup \semtrue{\psi}{^\eta} &  
   & \semfalse{\varphi \vee \psi}{^\eta} &=&\semfalse{\varphi}{^\eta} \cap \semfalse{\psi}{^\eta}\\

\semtrue{\lozenge \varphi}{^\eta} &=&  \tilde\lozenge \semtrue{\varphi}{^\eta} &
   & \semfalse{\lozenge \varphi}{^\eta} &=& \tilde\Box \semfalse{\varphi}{^\eta} \\

\semtrue{\Box \varphi}{^\eta} &=&  \tilde\Box \semtrue{\varphi}{^\eta} &
   & \semfalse{\Box \varphi}{^\eta} &=& \tilde\lozenge \semfalse{\varphi}{^\eta} \\

\semtrue{\mu X. \varphi}{^\eta} &=&  \mu U. \semtrue{\varphi}{^{\eta[X := U]}} &
   & \semfalse{\mu X. \varphi}{^\eta} &=& \nu U. \semfalse{\varphi}{^{\eta[X := U]}} \\

\semtrue{\nu X. \varphi}{^\eta} &=&  \nu U. \semtrue{\varphi}{^{\eta[X := U]}} &
   & \semfalse{\nu X. \varphi}{^\eta} &=& \mu U. \semfalse{\varphi}{^{\eta[X := U]}}
\end{array}
\]
Here, we used the following two abbreviations:
\[
\begin{array}{lcl}
\tilde\lozenge U = \{s \in S ~|~
\exists A \subseteq S.~ s \rel{R^+} A \wedge A \subseteq U \}
&\qquad \text{and}\qquad&
\tilde\Box U = \{s \in S ~|~ \forall t \in S.~ s \rel{R^-} t \Rightarrow
t \in U \}
\end{array}
\]
In case the formula $\varphi$ is \emph{closed}, its semantics is
independent of the environment $\eta$, and we drop $\eta$ from the
semantic brackets.  A closed formula $\varphi$ is \emph{true} in
a state $s \in S$, denoted $\tuple{M,s} \models^{SIS} \varphi$ (or
simply $\tuple{M,s} \models \varphi$ if the setting is clear from the
context) if $s \in \semtrue{\varphi}{}$; $\varphi$ is \emph{false} in $s$,
denoted $\tuple{M,s} \not\models^{SIS} \varphi$ ($\tuple{M,s} \not\models
\varphi$) if $s \in \semfalse{\varphi}{}$ and it is \emph{unknown} in $s$
otherwise.

\end{definition}
Note that if $M$ is a Kripke structure, it always holds that either
$\tuple{M,s} \models \varphi$ or $\tuple{M,s} \not\models \varphi$,
and the satisfaction relation coincides with the usual semantics for
the $\mu$-calculus.

By abuse of notation, we make $\mGTS$ also stand for the abstraction
formalism $\tuple{\mKS, \mGTS, \mixsim, \models^{SIS}}$ Likewise, we
make $\mKMTS$ stand for the abstraction formalism $\tuple{\mKS, \mKMTS,
\mixsim, \models^{SIS}}$.\\

In several works on abstraction \cite{WGC:11,SG:08}, the authors restrict
themselves to a particular abstraction setting of the form $\tuple{C,
\rho,S}$, in which we are given two sets of concrete and abstract objects
and a fixed abstraction relation; properties of interest (precision,
see~\cite{SG:08}, consistency but also expressiveness, see~\cite{WGC:11})
are then analysed relatively to such a fixed abstraction relation. The
abstraction setting we define below formalises the approach of Wei
\emph{et al.}~\cite{WGC:11}; it differs from \cite{SG:08} in that $C$
is assumed to be a set of concrete states without a structure, whereas
in \cite{SG:08} it is assumed to be a Kripke Structure.

\begin{definition}
 An \emph{abstraction setting} is a pair $\tuple{\calf,
 \tuple{C,\rho,S}}$, such that $\calf$ is an abstraction formalism, $C$
 and $S$ are sets, and $\rho \,\subseteq\, C \times S$ is an approximation
 relation.

\end{definition}

\section{On the Expressiveness of GTSs and KMTSs}
\label{sec:expressiveness}

In Section~\ref{subsec:expressiveness} we first define the notion
of expressiveness, based on the definition of semantics of an abstract model that is used by \eg Godefroid and Jagadeesan
in~\cite{GJ:03}.  Then, in Section~\ref{subsec:cont_expr}, we formalise
the notion of expressiveness used by Wei \emph{et al} in~\cite{WGC:11}.

\subsection{Expressiveness}
\label{subsec:expressiveness}

The expressiveness of an abstraction framework characterises classes
of concretisations (refinements, implementations), that one is able to
capture using abstract models from  the framework. The definition is
clear and intuitive; however, it still depends on the choice of what
constitutes the ``proper'' concrete semantics of an abstract model. In
most cases \cite{GJ:03,GC:05,LNW:07}, the latter is defined simply as the
class of all those Kripke structures that refine a given model. Formally,
the definition is as follows:

\begin{definition}
\label{def:concret}
Given an abstraction formalism ${\cal F}\,=\,\absform$ the class of \emph{concretisations} of a state $\spair{M}{s}$, where $M \in {\cal M}$ is defined as: 
\begin{center}
$\concret{M}{s}{\calf} \,\eqdef\,  \{\spair{K}{s_K} \in {\cal C} \,\mid\, \spair{K}{s_K} \metaref \spair{M}{s}\}$
\end{center}
\end{definition}
Having fixed the semantics of an abstract model, we may now define the
corresponding notion of expressiveness.

\begin{definition}
\label{def:express}
Let $\calf_1 = \tuple{\calc,\calm_1,\metaref_1,\absmodels_1}$
and $\calf_2 = \tuple{\calc,\calm_2,\metaref_2,\absmodels_2}$ be
two abstraction formalisms. We say that ${\cal F}_2$ is \emph{more
expressive than ${\cal F}_1$}, denoted $\calf_1 \sqsubseteq_{ex}
\calf_2$, if for all $\spair{M_1}{s_1} \in {\cal M}_1$ there exists
$\spair{M_2}{s_2} \in {\cal M}_2$ such that $\concret{M_1}{s_1}{\calf_1}
= \concret{M_2}{s_2}{\calf_2}$. We say that ${\cal F}_2$ is
\emph{strictly more expressive than ${\cal F}_1$}, denoted $\calf_1
\sqsubset_{ex} \calf_2$, if $\calf_1 \sqsubseteq_{ex} \calf_2$ and
there is some $\spair{M_2}{s_2} \in {\cal M}_2$ such that
$\concret{M_2}{s_2}{\calf_2} \not= \concret{M_1}{s_1}{\calf_1}$ for
every $\spair{M_1}{s_1} \in {\cal M}_1$.

\end{definition}

\begin{theorem}
\label{thm:more_expr}
 GTSs are strictly more expressive than KMTSs, \ie $\mKMTS
 \sqsubset_{ex} \mGTS$.

\end{theorem}

\begin{proof} GTSs subsume KMTSs syntactically, and both abstraction
formalisms use the same approximation relation. Therefore they are at
least as expressive as KMTSs, \ie $\mKMTS \sqsubseteq_{ex} \mGTS$. We will
prove that they are strictly more expressive than KMTSs by showing that a
certain GTS cannot be matched by a semantically equivalent KMTS. 

Consider the GTS $G$ depicted below, where $L(s) = \{a\}$ and $L(q)
= \{\neg a\}$; the labelling is denoted with square brackets, the
must hypertransitions are depicted by solid transitions pointing to a rectangle
containing one or more states, whereas dashed arrows are the may transitions.

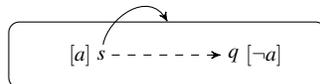
\tikzset{every loop/.style={min distance=10mm, looseness=10}}
\begin{center}
\vspace{-10pt}

\tikzset{my loop/.style={->,to path={ .. controls + (80:0.5) and +(120:0.5) .. (\tikztotarget) \tikztonodes}}
}

\begin{tikzpicture}[->,>=stealth', auto, node distance=25pt]
 \tikzstyle{every state}=[text=black,inner sep=1pt, minimum size=3pt, draw=none]

%mystate/.style ={text=black, font=\scriptsize, minimum size=1pt}

%[label distance=1cm]30:label}

%\node[state] (t0){};% [below right of=s0,xshift=120pt] 
\node[label={[label distance=-4pt]180:\scriptsize $[a]$},state] (t1) {\scriptsize $s$};
\node[state,shape=rectangle,rounded corners, minimum width=120pt,minimum height=25pt, fill=none, draw] (t01) [right of=t1]  {};
\node[label={[label distance=-4pt]0: \scriptsize $[\neg a]$},state] (t2) [right of=t01] {\scriptsize $q$};

\draw[->]
%(t0)  edge (t01)
(t1) edge [my loop]  (t01.north)
(t1) edge[dashed] (t2)
%(t2) edge [loop right] (t2)
;
\end{tikzpicture}
\vspace{-3pt}

\end{center}
Suppose, towards a contradiction, that for a certain $M \in \mMTS$,
and a certain state $s_M$ of $M$, we have $\concret{M}{s_M}{\mMTS} =
\concret{G}{s}{\mGTS}$.
First, observe that for every state $q_M$, reachable from $\tuple{M,s_M}$,
one of the following holds:
\begin{compactenum}
 \item $a \in L(q_M)$ and there is a must transition starting at $q_M$
 \item $\neg a \in L(q_M)$ and there is no transition starting at $q_M$
\end{compactenum}
To prove the above, one can show that these two properties, which hold
on all (two) reachable states of $G$, are preserved ``backwards''
by mixed simulation, so they hold in all states of all refinements
of $\tuple{G,s}$. On the other hand, if one of them did not hold
in some state of $M$, then $M$ would have a refinement outside
$\concret{G}{s}{\mGTS}$.\\

\noindent
The next observation is that both states $p_1$ and $p_2$ depicted below
are valid concretisations of $\tuple{G,s}$:
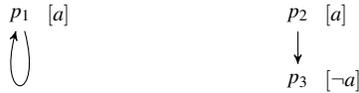
\begin{center}
%\vspace{-5pt}

\begin{tikzpicture}[->,>=stealth', auto, node distance=25pt]

 \tikzstyle{every state}=[text=black,inner sep=1pt, minimum size=3pt,draw=none]
%mystate/.style ={text=black, font=\scriptsize, minimum size=1pt}
 
\node[label=right:{\scriptsize $[a]$},state] (p1) {\scriptsize $p_1$};

\draw[->]
(p1) edge [loop below] (p1)
;

\node[label=right:{\scriptsize $[a]$},state] (p2) [right of=p1,xshift=80pt] {\scriptsize $p_2$};
\node[label=right:{\scriptsize $[\neg a]$},state] (p3) [below of=p2] {\scriptsize $p_3$};

\draw[->]
(p2)  edge (p3)
;
\end{tikzpicture}
\vspace{-5pt}
\end{center}
Let us now focus on $M$. Since the deadlocked process is not a
concretisation of $\tuple{G,s}$, there has to be a must transition
starting in $\tuple{M,s_M}$ (case 1 above), say, $\tuple{M,s_M}
\rmust \tuple{M,s'_M}$. Consider two cases:

\begin{compactenum}
 \item If there is a must transition starting at $s'_M$, then $p_2$ is not a concretisation of $\tuple{M,s_M}$, a contradiction.
 \item If $\neg a \in L(q_M)$, then $p_1$ is not a concretisation of $\tuple{M,s_M}$, a contradiction.
 \end{compactenum}
Both cases lead to the desired contradiction.
\end{proof}
Note that in the above proof we did not make any assumption on the
finiteness of $M$. That means that the GTS $G$ cannot be semantically
matched even by an infinite-state KMTS, \ie our expressiveness result is
rather robust.

\subsection{Contextual Expressiveness}
\label{subsec:cont_expr}

We already mentioned that the definition of expressiveness crucially
depends on the choice of concretisations of an abstract model. Apart from
taking the entire class, we may consider some restricted subclasses,
for instance only deterministic implementations \cite{FE:09}. 

From the model checking perspective, the most important alternative
definition was given in \cite{WGC:11}. There, properties of an abstraction
framework are always considered in the context of a specific abstraction,
namely a relation $\rho \subseteq C \times S$  for some particular sets
$C$ and $S$ of concrete and abstract objects, see the definition of the
\emph{abstraction setting}.

Given a model whose states consist of elements of $S$, we only consider
concretisations from the given abstraction setting (based on $\rho$). The
definition below is based on \cite{WGC:11}.

\begin{definition}
\label{def:concret-rel}
Given a fixed abstraction setting $\tuple{\calf, \tuple{C,\rho,S}}$,
we define the set of concretisations of a state, $\spair{M}{s}$, where $M
\in {\cal M}$, and $\states{M} \subseteq S$, \emph{in the context of
$\rho$}, as:

\begin{center}
$\concretfix{M}{s}{\rho}{\calf} \,\eqdef\, \{\spair{K}{s_K} \in C \,\mid\, \states{K} \subseteq C, \, \spair{K}{s_K} \rho \spair{M}{s} $ and $\rho$ is a refinement relation in the sense of $\metaref$ between $K$ and $M$\}
\end{center}
\end{definition}
With the semantics of abstract models restricted to the specific instance
of abstraction, we automatically obtain another 
notion of expressiveness. The one
given in \cite{WGC:11} is somewhat ambiguous: 
\begin{quote}``Two partial modelling
formalisms are \emph{expressively equivalent} if and only if for every
transition system $M$ from one formalism, there exists a transition
system $M'$ from the other, such that $M$ and $M'$ are semantically
equivalent.''
\end{quote}
The semantic equivalence to which the quote refers means having equal
sets of concretisations in the context of a specific abstraction (in
our terms, the equality of the $\concretfix{M}{s}{\rho}{\calf}$ sets).

For $M$, we can assume a specific setting $\tuple{\calf,
\tuple{C,\rho,S}}$. However, it is not immediately obvious whether there
are restrictions on the refinement relation $\rho'$ with which $M'$ is
supposed to match $M$: can it be arbitrary, or should it be in some way
related to $\rho$? We believe the best way is to allow for an arbitrary
abstraction setting $\tuple{C,\rho',S'}$ for $M'$, with one practical
assumption that whenever $S' \cap S \neq \emptyset$, then $\rho'$
conservatively extends $\rho$ on $S \cap S'$, \ie  $\{ s_C \,\mid\,
\exists s \in S' \cap S. \, s_C \rel{\rho'} s\} =\{ s_C \,\mid\,
\exists s \in S' \cap S. \, s_C \rel{\rho} s\}$.

The following definition of expressiveness, which we dub \emph{contextual
expressiveness}, so as to avoid confusion with the notion of
expressiveness we considered in the previous section, is based on
\cite{WGC:11}; we clarify their notion by making the dependence on the
abstraction context explicit.

\begin{definition}
\label{def:cont_express}
Let $\calf_1 = \tuple{\calc,\calm_1,\metaref_1,\absmodels_1}$
and $\calf_2 = \tuple{\calc,\calm_2,\metaref_2,\absmodels_2}$ be
two abstraction formalisms. Then ${\cal F}_2$ is \emph{contextually more
expressive than ${\cal F}_1$}, denoted ${\cal F}_1 \sqsubseteq_{cex}
{\cal F}_2$ if for all abstraction settings $\tuple{\calf_1,
\tuple{ C, \rho, S}}$ and all model-state pairs $\spair{M_1}{s_1}
\in {\cal M}_1$ with $\states{M_1} \subseteq S$, there is
an abstraction setting $\tuple{\calf_2,\tuple{C,\rho',S'}}$
$\spair{M_2}{s_2} \in {\cal M}_2$, with $\states{M_2}
\subseteq S'$ such that $\concretfix{M_1}{s_1}{\calf_1}{\rho} =
\concretfix{M_2}{s_2}{\calf_2}{\rho'}$.

\end{definition}

\begin{theorem}
 $\mGTS$ and $\mMTS$ are contextually equally expressive 
 \ie $\mGTS \sqsubseteq_{cex} \mMTS$ and $\mMTS \sqsubseteq_{cex}
 \mGTS$, see~\cite{WGC:11}.
\end{theorem}
\newcommand{\gtok}{\textsc{GtoK}\xspace}
\newcommand{\gtoktr}[1]{\ensuremath{\textsc{GtoK}}(#1)}
We believe it is at this point instructive to explain why the GTS
we used to prove strictness of expressiveness in the previous section
does not work in the setting of contextual expressiveness.  For this,
we consider the transformation \gtok defined in~\cite{WGC:11}, which,
given a GTS in a specific abstraction setting $\tuple{C,\rho,S}$, produces
a contextually equally expressive KMTS. 

We first show that the application of \gtok does not produce an
expressively equivalent (in our sense) KMTS.  For this, we need to define
an abstraction setting. Let us consider a very simple one, with concrete
and abstract sets consisting of two elements, namely $C = \{s_C, q_C\}$
and $S = \{s,q\}$, and the description relation defined as $\rho =
\{(s_C,s),\, (q_C,q)\}$.

Now consider the GTS $G$ from the proof of Thm.~\ref{thm:more_expr}. The
set of concretisations of the state $\tuple{G,s}$ in the context of
$\rho$, $\concretfix{G}{s}{\mGTS}{\rho}$, consists of the following
three models:

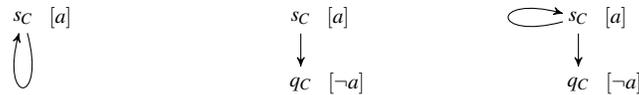
\begin{center}
%\vspace{-5pt}

\begin{tikzpicture}[->,>=stealth', auto, node distance=25pt]
 \tikzset{every loop/.style={min distance=10mm, looseness=10}}
 \tikzstyle{every state}=[text=black,inner sep=1pt, minimum size=3pt,draw=none]

%mystate/.style ={text=black, font=\scriptsize, minimum size=1pt}
 
\node[label=right:{\scriptsize $[a]$},state] (sc1) {\scriptsize $s_C$};

\draw[->]
(sc1) edge [loop below] (sc1)
;

\node[label=right:{\scriptsize $[a]$},state] (sc2) [right of=sc1,xshift=80pt] {\scriptsize $s_C$};
\node[label=right:{\scriptsize $[\neg a]$},state] (sc21) [below of=sc2] {\scriptsize $q_C$};

\draw[->]
(sc2)  edge (sc21)
;

\node[label=right:{\scriptsize $[a]$},state] (sc3) [right of=sc2,xshift=80pt] {\scriptsize $s_C$};
\node[label=right:{\scriptsize $[\neg a]$},state] (sc31) [below of=sc3] {\scriptsize $q_C$};

\draw[->]
(sc3) edge [loop left] (sc3)
(sc3)  edge (sc31)
;

\end{tikzpicture}
\vspace{-5pt}
\end{center}
Next, we consider the transformation \gtok, and apply it to the
GTS $G$ from the proof of Thm.\ref{thm:more_expr}. The starting point
of the transformation is the KMTS consisting of the same states as the
original GTS, with every original standard transition (with one state
as a target). Then for every true hypertransition $p \rel{R}^+ A$, \gtok
introduces a new state that mimics the target of the hypertransition $p_A$
such that $\gamma(p_A) = \bigcup_{p' \in A} \gamma(p')$. Finally, the
proper may transitions are added to the newly introduced states.
This results in the following KMTS for our GTS $G$:

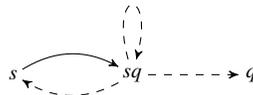
 \begin{center}
 \begin{tikzpicture}[->,>=stealth', auto, node distance=25pt]
 \tikzset{every loop/.style={min distance=10mm, looseness=10}}
 \tikzstyle{every state}=[text=black,inner sep=1pt, minimum size=3pt,draw=none]

%mystate/.style ={text=black, font=\scriptsize, minimum size=1pt}

\node[state] (s) {\scriptsize $s$};
\node[state] (sq) [right of=s,xshift=20pt] {\scriptsize $sq$};
\node[state] (q) [right of=sq,xshift=20pt] {\scriptsize $q$};

\draw[->]
(s)  edge[bend left,solid] (sq)
(sq)  edge[loop above, dashed] (sq)
(sq)  edge[bend left,dashed] (s)
(sq)  edge[dashed] (q)
;

\end{tikzpicture}
\vspace{-5pt}
\end{center}
It is not too hard to check that $\concretfix{G}{s}{\mGTS}{\rho}
= \concretfix{\gtoktr{G}}{s}{\mMTS}{\rho'}$, where $\rho'$ is the
abstraction relation resulting from the transformation, \ie  $\rho'
= \{(s_C,s),\, (q_C,q), \,(s_C,sq),\,(q_C,sq)\}$. This means that $G$
does not show that $\mGTS$ is contextually strictly more expressive than
$\mMTS$ (of course, this could never be the case as the transformation
$\gtok$ is used to prove that an \emph{arbitrary} GTS can be converted to
a contextually equivalent KMTS).

On the other hand, the KMTS above is not expressively equivalent to $G$,
because for instance the following Kripke structure is a concretisation
of $\gtoktr{G}$:
\begin{center}
\vspace{-10pt}
 \begin{tikzpicture}[->,>=stealth', auto, node distance=25pt]
 \tikzstyle{every state}=[text=black,inner sep=1pt, minimum size=3pt,draw=none]

%mystate/.style ={text=black, font=\scriptsize, minimum size=1pt}

\node[label=left:{\scriptsize $[a]$},state] (q1) {\scriptsize $q_1$};
\node[label=right:{\scriptsize $[a]$},state] (q2) [right of=sq,xshift=10pt] {\scriptsize $q_2$};
\node[draw=none] (q3) [above of=q1, yshift=-10pt] {};

\draw[->]
(q1)  edge[solid] (q2)
%(q2)  edge[loop right, dashed] (s2)
;

\end{tikzpicture}
\vspace{-5pt}
\end{center} 

From the fact that $\mGTS$ and $\mMTS$ are contextually equally
expressive but $\mGTS$ are strictly more expressive than
$\mMTS$, it follows that contextual expressiveness does not imply
expressiveness in general. In fact, if we consider arbitrary abstraction
frameworks, these two notions are incomparable. The reason is
that a certain framework may only allow maximal refinement relations (\eg
only maximal mixed simulations) between concrete and abstract models,
which will make it less contextually expressive than if we allow all
possible relations.

\begin{theorem}
Expressiveness and contextual expressiveness do not imply one another.

\end{theorem} 

\begin{proof}
For an example showing that contextual expressiveness does not 
imply expressiveness we can use the example of GTSs and KMTSs.

To prove that the implication does not hold in the other direction
as well, we can use frameworks based on the same class of models, but
different refinement relations. Suppose that $\calf_1$ is a standard
KMTS framework with mixed simulation as a meta-refinement relation,
and $\calf_2$ differs from $\calf_1$  in that it allows only those
mixed simulations that are maximal, for a given concrete and abstract
model. Clearly, $\calf_1$ and $\calf_2$ are expressively equivalent. On
the other hand, they are not equivalent when it comes to contextual
expressiveness.

In order to see this, consider a very simple KMTS $M^{all}$ depicted
below, and its concretisation $K$. Note that both states of $M^{all}$ abstract the entire class of Kripke structures.

 \begin{center}
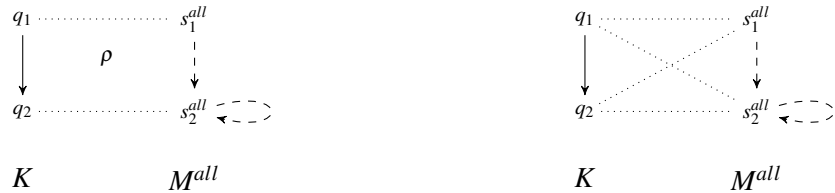
\begin{tabular}{lll}
 \begin{tikzpicture}[->,>=stealth', auto, node
 distance=25pt] \tikzset{every loop/.style={min distance=10mm,
 looseness=10}} \tikzstyle{every state}=[text=black, font=\scriptsize, inner sep=1pt,
 minimum size=3pt,draw=none]

%mystate/.style ={text=black, font=\scriptsize, minimum size=3pt}

%\tikzstyle{every node} = [circle, fill=gray!30]

\node[state] (q1) {\scriptsize $q_1$};
\node[state] (q2) [below of=q1, yshift=-10pt] {\scriptsize $q_2$};
\node [below of = q2] {$K$};
%\node[state][label=below:{\scriptsize $s^{all}_2$},state] (s2) [right of=s1,xshift=40pt] {};

\draw[->]
(q1)  edge[solid] (q2)
;

\node[state] (s1) [right of = q1, xshift=40pt] {\scriptsize $s^{all}_1$};
\node[state] (s2) [below of=s1, yshift=-10pt] {\scriptsize $s^{all}_2$};
\node [below of = s2] {$M^{all}$};
\node[below of = q1, xshift = 32pt, yshift = 10pt]{\scriptsize $\rho$};
%\node[state][label=below:{\scriptsize $s^{all}_2$},state] (s2) [right of=s1,xshift=40pt] {};

\draw[->]
(s1)  edge[dashed] (s2)
(s2)  edge[loop right, dashed] (s2)
;

\draw[-]
(q1) edge[dotted] (s1)
(q2) edge[dotted] (s2)
;

\end{tikzpicture}
&
\hspace{80pt}
&
\begin{tikzpicture}[->,>=stealth', auto, node
 distance=25pt] \tikzset{every loop/.style={min distance=10mm,
 looseness=10}} \tikzstyle{every state}=[text=black, font=\scriptsize, inner sep=1pt,
 minimum size=3pt,draw=none]

%mystate/.style ={text=black, font=\scriptsize, minimum size=3pt}

%\tikzstyle{every node} = [circle, fill=gray!30]

\node[state] (q1) {\scriptsize $q_1$};
\node[state] (q2) [below of=q1, yshift=-10pt] {\scriptsize $q_2$};
\node [below of = q2] {$K$};
%\node[state][label=below:{\scriptsize $s^{all}_2$},state] (s2) [right of=s1,xshift=40pt] {};

\draw[->]
(q1)  edge[solid] (q2)
;

\node[state] (s1) [right of = q1, xshift=40pt] {\scriptsize $s^{all}_1$};
\node[state] (s2) [below of=s1, yshift=-10pt] {\scriptsize $s^{all}_2$};
\node [below of = s2] {$M^{all}$};
%\node[below of = q1, xshift = 34pt, yshift = 16pt]{$\rho'$};
%\node[state][label=below:{\scriptsize $s^{all}_2$},state] (s2) [right of=s1,xshift=40pt] {};

\draw[->]
(s1)  edge[dashed] (s2)
(s2)  edge[loop right, dashed] (s2)
;

\draw[-]
(q1) edge[dotted] (s1)
(q1) edge[dotted] (s2)
(q2) edge[dotted] (s2)
(q2) edge[dotted] (s1)
;

\end{tikzpicture}
\end{tabular}
\vspace{-5pt}
\end{center} 
 
Let us now define an abstraction setting, $\tuple{\calf_1, \tuple{ \{q_1,q_2\}, \rho, \{s^{all}_1, s^{all}_2\}}}$, where $\rho(q_1) = s^{all}_1$ and $\rho(q_2) = s^{all}_2$ (on the left picture above). It is a valid abstraction setting; indeed, $\rho$ is a mixed simulation, so it is a proper refinement in the framework $\calf_1$. However, $\rho$ is not a refinement in $\calf_2$, because is is contained the relation $\{p_1, p_2\} \times \{s_1,s_2\}$ (depicted above on the right), which is the only maximal mixed simulation between $K$ and $M^{all}$. Since this maximal refinement is the only one that we are allowed to consider as a candidate for $\rho'$ to match $\rho$ in concretisations of $s^{all}_1$ and $s^{all}_2$, we observe that it is impossible to find a $\rho'$ with the property that $\concretfix{M^{all}}{s^{all}_1}{\calf_1}{\rho} = \concretfix{M^{all}}{s^{all}_1}{\calf_2}{\rho'}$. Hence $\calf_1$ and $\calf_2$ are not contextually expressively equivalent.
\end{proof}

\section{Completeness}
\label{sec:completeness}

\newcommand{\compl}[1]{\textbf{compl}(#1)}
\def\complpre{\sqsubseteq_{cmp}}
\def\expre{\sqsubseteq_{ex}}

Arguably the most important property of any abstraction formalism is
its degree of completeness. We will first provide a formal measure of
completeness; for a given formalism, it is defined as the set of formulae
that can be proved with a finite abstract model.

\begin{definition} 

Given an abstraction formalism ${\cal F} = \absform $. The 
\emph{completeness set} of ${\cal F}$, denoted $\compl{{\cal F}}$ is
defined as follows:
%$\compl{{\cal F}}{\metaref_{sa}}{\models_{sem}}$
\[
\begin{array}{lll}
\compl{{\cal F}} &\,\eqdef\,& \{\varphi \in \lmu \,\mid\, \varphi \text{ is satisfiable and } \forall \tuple{K,s_K} \in {\cal C}. \, \tuple{K,s_K} \models \varphi \,\implies\,\\
& & \exists \tuple{M,s_M} \in {\cal M}.\, \tuple{K,s_K} \metaref \tuple{M,s_M} \,\wedge\, \tuple{M,s_M} \models^{\alpha} \varphi \}
\end{array}
\]
Furthermore, we define the \emph{relative completeness preorder} between 
abstraction formalisms ${\cal F}_1, {\cal F}_2$ as:
\begin{center}
${\cal F}_1 \complpre {\cal F}_2 \iffdef \compl{{\cal F}_1} \subseteq \compl{{\cal F}_2}$
\end{center}
\end{definition}
Note that we have assumed that all abstract models in ${\cal M}$ are finite.

\subsection{Expressiveness and completeness}

We start by showing some simple facts concerning the relationship
between expressiveness and completeness. One can observe an implication
in one direction, if we assume the setting with thorough semantics,
\ie when the satisfaction of a formula by an abstract model depends
on whether it is satisfied by all its concretisations. (note that this
semantics differs from the standard inductive semantics we introduced
in Section~\ref{sec:preliminaries}).

\begin{definition}
Given an abstraction formalism $\tuple{\calc,\calm,\metaref,\models^{t}}$,
we say that $\models^{t}$ is a \emph{thorough semantics} if it is
defined as (or satisfies):
 \begin{center}
  $\tuple{M,s} \models^{t} \varphi \,\iff\, \forall \tuple{K,s'} \in \concret{M}{s}{\calf}\,.\, \tuple{K,s'} \models \varphi$
 \end{center} 
\end{definition}

\begin{proposition}
Assuming thorough semantics, expressiveness containment implies
relative completeness containment. More precisely, if $\calf_1
=\tuple{\calc,\calm_1,\metaref_1,\models_{1}^t}$ and $\calf_2
=\tuple{\calc,\calm_2,\metaref_2,\models_{2}^t}$, where $\models_{1}^{t}$
and $\models_{2}^{t}$ are thorough semantics, then:
 \begin{center}
 ${\cal F}_1 \expre {\cal F}_2 \,\implies\, {\cal F}_1 \complpre {\cal F}_2$ 
 \end{center} 
\end{proposition}

\begin{proof}
Take any $\varphi \in \compl{\calf_1}$, and an arbitrary $\tuple{K,s_K}
\in \mKS$ such that $\tuple{K,s_K} \models \varphi$. From the
fact that $\varphi \in \compl{\calf_1}$ it follows that there
exists some $\tuple{M_1,s_1} \in \calm_1$ such that $\tuple{K,s_K}
\in \concret{M_1}{s_1}{\calf_1}$ and $\tuple{M_1,s_1} \models_1^{t}
\varphi$. Because of thorough semantics, we then have $\forall
\tuple{K',s'} \in \concret{M_1}{s_1}{\calf_1}\,.\,\tuple{K',s'} \models
\varphi$.

Since ${\cal F}_1 \expre {\cal F}_2$, there is a model
$\tuple{M_2,s_2} \in \calm_2$ such that $\concret{M_1}{s_1}{\calf_1} =
\concret{M_2}{s_2}{\calf_2}$. Because of this equality $\tuple{M_1,s_1}$
and $\tuple{M_2,s_2}$ satisfy the same formulae under thorough semantics,
and hence $\tuple{M_2,s_2} \models_2^t \varphi$. Since $K$ was chosen
arbitrarily, we obtain $\varphi \in \compl{\calf_2}$.
\end{proof}

The necessity of a thorough semantics of both formalisms can be understood
from the following counterexample that shows that, in general, equal
expressiveness does not imply equal completeness. Take as $\calf_1$ the
standard class of $\mKMTS$ with inductive semantics and as $\calf_2$
also KMTSs, but this time with a thorough semantics. Since inductive
semantics does not preserve thorough semantics, these two formalisms
have different relative completeness. 

Consider the following example, taken from \cite{GJ:03}; the formula
$\boxop p \wedge \neg \boxop q$ is \emph{false} on all refinements of
the KMTS depicted below, but it is unknown on the KMTS itself.

\begin{center}
\vspace{-7pt}

\begin{tikzpicture}[->,>=stealth', auto, node distance=19pt]
 \tikzstyle{every state}=[text=black,inner sep=1pt, minimum size=3pt,fill=black]

%mystate/.style ={text=black, font=\scriptsize, minimum size=3pt}

%\node[state] (t0){};% [below right of=s0,xshift=120pt] 
\node[label=left:{\scriptsize $[p,q]$},state] (t1) {};%[right of=t0,xshift=10pt] {};
%\node[state,shape=rectangle,rounded corners, minimum width=70pt,minimum height=18pt, fill=none] (t01) [right of=t1]  {};
\node[label=right:{\scriptsize $[\neg p, \neg q]$},state] (t2) [right of=t1,xshift=20pt] {};
\node [above of=t1] {};
\node [below of=t1] {};

\draw[->]
%(t0)  edge (t01)
%(t1) edge [my loop]  (t01.north)
(t1) edge[dashed] (t2)
%(t2) edge [loop right] (t2)
;
\end{tikzpicture}
\vspace{-10pt}
\end{center}

\begin{theorem}
 In general, more completeness does not imply more expressiveness, even in formalisms with thorough semantics.
\end{theorem}

\begin{proof} Consider two formalisms, consisting of simple KMTSs
depicted below (we assume a setting with mixed simulation, and inductive
or thorough semantics).
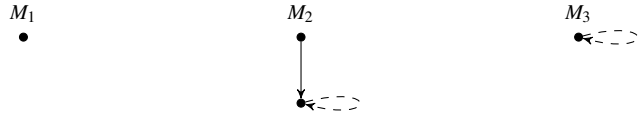
\begin{center}
\vspace{-20pt}
\begin{tikzpicture}[->,>=stealth', auto, node distance=25pt]
 \tikzset{every loop/.style={min distance=10mm, looseness=10}}
 \tikzstyle{every state}=[text=black,inner sep=1pt, minimum size=3pt,fill=black]

\node[label=above:{\scriptsize $M_1$},state] (s1) {};

%\draw[->]
%(p1) edge [loop right] (p1)
%;

\node[label=above:{\scriptsize $M_2$},state] (s2) [right of=s1,xshift=80pt] {};
\node[state] (s21) [below of=s2] {};

\draw[->]
(s2)  edge (s21)
(s21)  edge[loop right, dashed] (s21)
;

\node[label=above:{\scriptsize $M_3$},state] (s3) [right of=s2,xshift=80pt] {};
\node [above of=s1] {};
\node [below of=s3] {};

\draw[->]
(s3)  edge[loop right, dashed] (s3)
;

\end{tikzpicture}
\vspace{-5pt}
\end{center} 
Let $\calm_1 = \{M_1, M_2 \}$ and $\calm_2 = \{M_1,M_2,M_3\}$. We have:
$\compl{\calf_1} = \compl{\calf_2} = \{\true, \diamop \true, \boxop
\false\}$ (up to semantic equivalence). However, $\calf_1$ is strictly
less expressive, because it cannot express the class of all Kripke
Structures with a single model (as is the case with $\calf_2$ and $M_3$).
\end{proof}

\subsection{Characterisation of completeness sets for GTSs}
%semantic completion
\newcommand{\semcom}[1]{\ensuremath{[#1]_{\equiv}}}

\newcommand{\minmodel}[3]{\ensuremath{\textbf{minmodel}_{#3}(#1,#2)}}
\newcommand{\maxminmodel}[2]{\ensuremath{\textbf{maxminmodel}_{#2}(#1)}}

The GTS framework is known to be complete for the set of least fixpoint
free formulae, see~\cite{AGJ:04}. However, it is not known whether there
are other subsets of $\lmu$ for which GTSs are complete, too. In this
section, we attempt to provide a more accurate characterisation of
the completeness set of the GTS framework.

From hereon, we only consider GTSs or their subclasses with standard
inductive semantics of $\lmu$. By $\lmu^{\nu}$ we will denote the fragment
of $\mu$-calculus in which only the greatest fixpoint is allowed.

Let $\equiv$ denote the semantic equivalence of formulae, \ie $\varphi
\equiv \psi$ if for every GTS $M$, $M \models^{SIS} \varphi \,\iff\,
M \models^{SIS} \psi $. For any sublanguage ${\cal L}' \subseteq \lmu$,
we define the completion of ${\cal L}'$ with respect to semantic
equivalence as $\semcom{{\cal L}'} \eqdef \{\varphi \in \lmu \,\mid\,
\exists \varphi' \in {\cal L}'.\, \varphi \,\equiv\, \varphi' \}$.

\newcommand{\approxsyn}[2]{\ensuremath{\textbf{approx}(#1,#2)}}
\newcommand{\approxvar}[4]{\ensuremath{\textbf{approx-var}^{#2,#3}_{#4}(#1)}}
\newcommand{\approxsynname}{\textbf{approx }}
\newcommand{\approxvarname}{\textbf{approx-var }}

\newcommand{\unf}[2]{\ensuremath{\textbf{unfold}_{#2}(#1)}}

\def\muvars{{\vars}_{\mu}}
\def\nuvars{{\vars}_{\nu}}
\def\varord{\lhd}
\def\subf{\textbf{Sub}}

Before we prove our main result, we need to discuss certain technical
issues.  Inspired by the standard ``naive'' decision procedure for $\mu$-calculus, that allows one to compute the semantics of a formula using
approximants, we observe that, whenever a finite GTS satisfies a formula
$\varphi$, it is witnessed by a least-fixpoint free \emph{syntactic} approximant.
These syntactic approximants can be obtained by successive unfoldings of
$\mu$-variables of $\varphi$.

By $\vars(\varphi)$ we will denote the set of all variables occurring
in $\varphi$, and $\subf(\varphi)$ will denote all subformulae of
$\varphi$. We will call a formula $\varphi$ \emph{well-formed}, if for
every $X \in \vars(\varphi)$, there is at most one subformula $\sigma
X. \psi_X \in \subf(\varphi)$. For a well-formed formula $\varphi$,
the unfolding of a bound variable $X \in \vars(\varphi)$, denoted
by $\unf{X}{\varphi}$, is the formula $\psi_X$ such that $\sigma
X. \psi_X \in \subf(\varphi)$. The set of bound $\mu$-variables
(resp.\ $\nu$-variables) of a well-formed formula $\varphi$ is denoted
with $\muvars(\varphi)$ (resp.\ $\nuvars(\varphi)$).

We first provide an auxiliary notion, namely the approximant of a
formula with respect to \emph{one} fixed $\mu$-variable.

\begin{definition}	
Fix a closed formula $\psi \in \lmu$. Let $\varphi \in \subf(\psi)$ be
a subformula of $\psi$ and let $k \in \nat$ be a positive
natural number.  Let $\gamma$ be a word
over $\nat$; we will use it as a subscript in recursion variables to
ensure well-formedness by introducing unique names. We will use the
convention that the variables $X,Y$ range over the variables of the
original formula $\psi$; in the definition below, they are formally
identified with $X_{\epsilon}, Y_{\epsilon}$, where $\epsilon$ denotes
the empty word.  A \emph{simple $k$-th approximant} of $\varphi$ with
respect to $X \in \muvars$, denoted $\approxvar{\varphi}{X}{k}{\psi}$,
is defined as follows:
  \[
  \begin{array}{llll}
   \approxvar{X_{\gamma}}{X}{0}{\psi} &\,\eqdef\,& \false & \\
   \approxvar{\mu X_{\gamma}. \varphi}{X}{0}{\psi} &\,\eqdef\,& \false & \\   
   \approxvar{\varphi}{X}{k}{\psi}&\,\eqdef\,& \varphi & \varphi \in \{\true, \false\} \cup \lit\\
   \approxvar{\varphi_1 \oplus \varphi_2}{X}{k}{\psi}&\,\eqdef\,& 
   \approxvar{\varphi_1}{X}{k}{\psi}\oplus \approxvar{\varphi_2}{X}{k}{\psi} &  \oplus \in \{\wedge, \vee\}\\
   \approxvar{\star \,\varphi}{X}{k}{\psi}&\,\eqdef\,& \star \, \approxvar{\varphi}{X}{k}{\psi}&
   \star \in \{\diamop,\boxop\}\\
   \approxvar{\sigma Y_{\gamma}. \varphi}{X}{k}{\psi}&\,\eqdef\,& \sigma Y_{\gamma\,k}. \, \approxvar{\varphi}{X}{k}{\psi}
   &Y \neq X, \sigma \in \{\mu,\nu\}
   \\
   %\approxvar{X}{X}{k}{\psi}&\,\eqdef\,& \approxvar{\varphi_X}{X}{k-1} & \\
   \approxvar{X_{\gamma}}{X}{k}{\psi}&\,\eqdef\,& \approxvar{\unf{X}{\psi}}{X}{k-1}{\psi} & \\
   \approxvar{Y_{\gamma}}{X}{k}{\psi}&\,\eqdef\,& Y_{\gamma\,k} &  Y \neq X\\
   
  \end{array}
  \] 
\end{definition}
Note that in the above definition, the original formula $\psi$ is kept
as a parameter of the \approxvarname operator, so that we are able to
retrieve the unfolding of the recursion variable $X$ in $\psi$.

\begin{proposition}
For any well-formed formula $\psi \in \lmu$, $\mu X. \varphi_X
\in \subf(\psi)$ and any $k \in \nat$, the $\mu$-calculus formula 
$\approxvar{\mu
X. \varphi_X}{X}{k}{\psi}$ is well-formed.

\end{proposition} 

\begin{proof} We use induction on $k$.
\begin{compactitem}
\item Base case: trivial.
\item Induction. Observe that $\approxvar{\mu X. \varphi_X}{X}{k+1}{\psi}$
preserves the structure of $\varphi_X$ apart from the occurrences of
$X$, which are unfolded with the ``lower'' approximants. Since $\mu
X. \varphi_X$ is well-formed, the only place where duplicate definition
of recursion variables could occur are the unfoldings. However, the
indexing scheme guarantees that recursion variables defined in the lower
approximants have different names. Together with the inductive
assumption, this yields well-formedness.
\end{compactitem}
\end{proof}
We are now in a position to define the general syntactic approximants of
a $\mu$-calculus formula, which are constructed by unfolding every $\mu$
variable a finite number of times (specified for each variable by a
function $\alpha$), so that the resulting formula is least-fixpoint free.

\begin{definition}
For a given formula $\varphi \in \lmu$ and $\alpha : \muvars(\varphi) \longrightarrow \nat$ we define an approximant of $\varphi$, denoted with $\approxsyn{\varphi}{\alpha}$ as:
  \[
   \approxsyn{\varphi }{\alpha} \,\eqdef\,
   \left \{
   \begin{array}{ll}
   \approxsyn{\approxvar{\varphi}{X}{\alpha(X)}{\varphi}}{\alpha} &
   \text{if $X_{\gamma} \in \muvars(\varphi)$ for some $\gamma \in \nat^{*}$ and}\\
   & \text{$\neg \exists \, Y_{\gamma'} \in \muvars(\varphi).\,$}
    \text{$\mu X_{\gamma} . \varphi ' \in \subf(\unf{Y_{\gamma'}}{\varphi})$}\\
   \varphi & \text{if $\muvars(\varphi) = \emptyset$ }   
   \end{array}
   \right.
  \]  
 \end{definition}

Intuitively, as long as there are $\mu$-variables in the formula, % (which are of the form $Y_{\gamma}$, where $Y$ was a $\mu$-variable of the original formula),
we take one (say $X$) such that for some $\gamma \in \nat^{*}$,  $X_{\gamma}$ has the outermost occurrence and apply the \approxvarname operator with $X$ as a parameter, by which we obtain a formula that does not contain any variable of the form $X_{\gamma'}$ for any $\gamma' \in \nat^{*}$.

 \begin{proposition} For any well-formed formula $\psi \in
 \lmu$ and any $\alpha : \muvars(\varphi) \longrightarrow \nat$,
 $\approxsyn{\psi}{\alpha}$ is a well-formed and least fixpoint-free
 $\mu$-calculus formula.
\end{proposition}

 \begin{proposition} \label{prop:witness}
  For any GTS $M$, $M \models \varphi \,\iff\, M \models
  \approxsyn{\varphi}{\alpha}$, where $\alpha(X) = |M|$ for every $X
  \in \muvars(\varphi)$.
 \end{proposition}
From the existence of least-fixpoint free approximants we can deduce the
following fact: for every formula containing a ``true'' least fixpoint
property $\varphi \in \lmu \setminus \semcom{\lmu^{\nu}}$, and for an
arbitrarily large number $n \in \nat$, we can always find a concrete
structure, which needs a GTS of size at least $n$ to prove $\varphi$ using
a GTS abstraction. We will use this fact to prove that, assuming that our
concrete structures are the most general class $\mKS$, the completeness
set of GTS is exactly the set of least-fixpoint free formulae. First,
we introduce some auxiliary notation, needed to facilitate proving this
result.

\begin{definition}
Let $\calf = \absform$ be an abstraction formalism, and $\varphi \in
\lmu$ a formula. For an arbitrary Kripke Structure $K \in \mKS$ 
we define 
\[ \begin{array}{lll}
 \minmodel{\varphi}{K}{\calf} & \,\eqdef\, & \left \{ \begin{array}{ll}
  inf \, \{|M| \,\mid\, M \in \calm \wedge K \metaref M \models^{\alpha}
  \varphi  \} & \text{ if } \varphi \in \compl{\calf}\\
 \infty & \text{ otherwise } \end{array} \right.  \\
  \maxminmodel{\varphi}{\calf} & \,\eqdef\,&  sup\,
  \{\minmodel{\varphi}{K}{\calf} \,\mid\, K \in \mKS \wedge K \models
  \varphi\} \end{array}
 \]

\end{definition}

\begin{lemma}
\label{lem:maxmin}
 For any $\varphi \in \lmu \setminus \semcom{\lmu^{\nu}}$, we have $\maxminmodel{\varphi}{\mGTS} = \infty$.
\end{lemma}

\begin{proof}
Suppose, towards a contradiction, that for some $\varphi \in \lmu \setminus
\semcom{\lmu^{\nu}}$, $\maxminmodel{\varphi}{\mGTS} = n$ for
some $n \in \nat$. Then for every Kripke Structure $K$ satisfying
$\varphi$ there is a GTS $M$ with a size at most $n$ such that $K
\mixsim M \models^{SIS} \varphi$. From Prop.~\ref{prop:witness}
we know that $M \models^{SIS} \varphi$ implies that $M \models
\approxsyn{\varphi}{\alpha_{|M|}}$ (where $\alpha_{|M|}(X)
= |M|$ for all $X \in \muvars(\varphi)$), and hence $K \models
\approxsyn{\varphi}{\alpha_{|M|}}$. This is because of the fact that
mixed simulation preserves
properties of abstract models to concrete, \ie the soundness. From the
above observations we obtain that for every $K \in \mKS$, we have $K
\models \varphi \iff K \models \approxsyn{\varphi}{\alpha_{|M|}}$,
so $\varphi \,\equiv\, \approxsyn{\varphi}{\alpha_{|M|}}$. But
$\approxsyn{\varphi}{\alpha_{|M|}} \in \lmu^{\nu}$, hence $\varphi \in
\semcom{\lmu^{\nu}}$, a contradiction.
\end{proof}
We are now ready to give the exact characterisation of the completeness
set of GTS, in case concrete structures are the most general class $\mKS$.

\begin{theorem}
\label{thm:compl_gts_infi}
For the general class $\mKS$, where an arbitrary, possibly infinite number of initial states is allowed, the class of formulae for which GTSs are complete, is, up to semantic equivalence, exactly $\lmu^{\nu}$, i.e. $\compl{\langle \mKS, \mGTS, \mixsim, \models^{SIS} \rangle} = \semcom{\lmu^{\nu}}$.
\end{theorem}

\begin{proof}
 The fact that $\lmu^{\nu} \subseteq \compl{\langle \mKS, \mGTS, \mixsim,
 \models^{SIS} \rangle}$ follows from the known result that GTSs are
 complete for $\lmu^{\nu}$ \cite{AGJ:04}. To prove the inclusion in the other direction, we
 proceed by contradiction. Assume that $\varphi \in (\lmu \setminus
 \semcom{\lmu^{\nu}}) \cap \compl{\langle \mKS, \mGTS, \mixsim,
 \models^{SIS} \rangle}$.

 Since $\varphi$ is satisfiable, there is a KS $K_1$ such that $K_1
 \models \varphi$, and from $\varphi \in \compl{\langle \mKS, \mGTS,
 \mixsim, \models^{SIS} \rangle}$ there exists a model $M_1 \in \mGTS$
 of size $n_1$ such that $K_1 \mixsim M_1 \models^{SIS} \varphi$. Assume
 that $M_1$ is the smallest GTS with this property. Because $\varphi \in
 (\lmu \setminus \semcom{\lmu^{\nu}})$, we know from Lem.~\ref{lem:maxmin}
 that $\maxminmodel{\varphi}{\mGTS} = \infty$, so there is some KS $K_2$,
 for which the corresponding smallest GTS $M_2$ proving $\varphi$ has size
 $n_2 > n_1$. We can continue this construction ad infinitum, obtaining
 a sequence ($K_i$,$M_i$,$n_i$), such that  $K_i \models \varphi$, $K_i
 \mixsim M_i \models^{SIS} \varphi$, $M_i$ is the smallest GTS proving
 $\varphi$ on $K_i$ and $n_i = |M_i|$.

 Let us now define $K = \bigcup_{i=1}^{\infty} K_i$ with $S^0 =
 \bigcup_{i=1}^{\infty} S^0_i$. Since for all $K_i \models \varphi$, we
 have for all $i \in \nat$, $s \in S^0_i$ $s \models \varphi$, therefore
 $K \models \varphi$. Since GTSs are complete for $\varphi$, there is
 some $M \in \mGTS$ such that $K \mixsim M \models^{SIS} \varphi$. But
 this means that for all $i \in \nat$ $\minmodel{\varphi}{K_i}{\mGTS}
 \leq |M|$, a contradiction.
\end{proof}
Note that in the proof of the above theorem, it is essential that one
is allowed to have an infinite set of initial states. If we restrict to
Kripke Structures with only finitely many initial states, then we obtain
a richer completeness set. We can consider two subcases, depending on
whether the concrete structures are finitely branching or not. While
we believe that it is difficult to provide an exact characterisation
of completeness sets in these cases, we can at least prove that the
completeness hierarchy is strict.

By $\mKS_{fi}$ we denote the subclass of $\mKS$ with finitely many
initial states; by $\mKS_{fb}$ we denote the subclass of $\mKS_{fi}$,
consisting of finitely branching structures.

\begin{theorem} We have both
$\compl{\langle \mKS, \mGTS, \mixsim, \models^{SIS} \rangle} \subset
\compl{\langle \mKS_{fi}, \mGTS, \mixsim, \models^{SIS} \rangle}$, and

\noindent
$\compl{\langle \mKS_{fi}, \mGTS, \mixsim, \models^{SIS} \rangle} \subset
\compl{\langle \mKS_{fb}, \mGTS, \mixsim, \models^{SIS} \rangle}$.
\end{theorem}

\begin{proof}
Obviously, smaller classes of concrete structures give rise to 
larger completeness sets in general. That the inclusions are strict, 
can be proved with the following counterexamples.
\begin{itemize}
  \item $\compl{\langle \mKS, \mGTS, \mixsim, \models^{SIS} \rangle} \subset
\compl{\langle \mKS_{fi}, \mGTS, \mixsim, \models^{SIS} \rangle}$:

By Thm.~\ref{thm:compl_gts_infi}, it suffices to show that there is at least
one ``true'' least
fixpoint formula that belongs to $\compl{\langle \mKS_{fi}, \mGTS,
\mixsim, \models^{SIS} \rangle}$. Consider the simple reachability
formula $\varphi = \mu X. P \vee \diamop X$ and suppose that $K
\models^{SIS} \varphi$. Then a state labelled with $P$ is reachable
from every initial state with a finite number of steps. For any state
$s$ of $K$, let $\textsc{steps}(s)$ denote the minimal path length
from $s$ to a $P$-labelled state, and let $n$ be the largest value of
$\textsc{steps}(s^0)$, among all initial states $s^0$. We can group
the states together according to the value of $\textsc{steps}(s)$,
and collapse all states with $\textsc{steps}(s) > n$ into one abstract
state. We complete the construction of the KMTS by adding must transitions
corresponding to decrementing $\textsc{steps}(s)$, and may transitions
whenever necessary. It is not too hard to see that such an KMTS mixed
simulates $K$ and allows to prove $\varphi$.

\item $\compl{\langle \mKS_{fi}, \mGTS, \mixsim, \models^{SIS} \rangle} \subset$ 
$\compl{\langle \mKS_{fb}, \mGTS, \mixsim, \models^{SIS} \rangle}$:

Consider a formula expressing that all computations terminate:
$\varphi = \mu X. \boxop X$. Firstly, observe that there exists an
infinitely branching model with one initial state, on which this
property holds but cannot be proved with a finite GTS. In its initial
state there is a choice between executing $n$ transitions for all $n \in
\nat$. Every execution path is bound to terminate, but since GTSs are
unable to compress the counting of steps, no GTS can finitely approximate
this Kripke Structure in a way that would allow to prove $\varphi$.

 \end{itemize}
\end{proof}

\section{Conclusions and future work}
\label{sec:conclusions}

Abstraction is often a key instrument for turning intractable model checking
problems into tractable problems. In the theoretical studies on abstraction
frameworks, the expressivity and completeness of the frameworks are the main
indicators of the power of the frameworks. 

A major problem is that the notion of expressiveness of an
abstraction framework is not defined unambiguously in the literature.
Wei \emph{et al}~\cite{WGC:11} established that using a certain notion
of expressiveness, the \GTS abstraction formalism is equally expressive
as the \KMTS abstraction formalism. In this paper, we showed that using
another common notion of expressiveness, occurring in the literature
on abstraction, the \GTS abstraction formalisms is actually strictly
more expressive than the \KMTS abstraction formalism. 

The same paper occasionally uses the notions of completeness and
expressivity interchangeably. Since both notions are defined differently,
we set out to investigate their relations. We proved that only under
specific conditions, there is a relation between completeness and
expressivity.

Finally, we studied the problem of completeness in more detail. We give
tighter characterisations of the completeness of the GTS framework.
Among others, we showed that GTSs are complete for exactly the least
fixpoint-free fragment of the $\mu$-calculus under the condition that
concrete models are Kripke Structures with potentially an infinite number
of initial states. We showed that these characterisations change when
imposing different requirements on the concrete models.

Several lines of future research are still open. For one, it would
be interesting to provide conditions under which expressiveness and
completeness are in some way related. A second avenue of research would
be investigating exact characterisations of completeness for subclasses
of Kripke Structures as concrete models. We expect that this is a very
difficult, yet challenging problem to solve.

\paragraph{Acknowledgements} The authors would like to thank the anonymous
reviewers for their detailed feedback which have helped to improve the paper.
This research has been partially funded by the Netherlands Organisation
for Scientific Research (NWO) under grant number 600.065.120 (the
\emph{Verification of Complex Hierarchical Systems} project).

\bibliographystyle{eptcs}
\bibliography{bibliography}

\end{document}